\documentclass[a4paper,11pt,reqno]{amsart}

\usepackage{amsthm}

\usepackage[colorlinks=true, linkcolor=blue, citecolor=blue, urlcolor=blue]{hyperref}

\title[Area-charge inequalities and rigidity of initial data sets]{Area-charge inequalities and rigidity of time-symmetric initial data sets}

\author[T.~Cruz]{Tiarlos Cruz}
\author[A.~Mendes]{Abraão Mendes}
\address{Institute of Mathematics, Federal University of Alagoas, 57072-970, 
	\linebreak\indent Maceió-AL, Brazil}
\email{\href{mailto: cicero.cruz@im.ufal.br}{cicero.cruz@im.ufal.br}}
\email{\href{mailto: abraao.mendes@im.ufal.br}{abraao.mendes@im.ufal.br}}

\subjclass[2020]{53A10, 53C24, 49Q05}

\keywords{Rigidity, Area-minimizing surfaces, $\mu$-bubbles, Area-charge inequality}

\numberwithin{equation}{section}

\renewcommand{\div}{\operatorname{div}}

\newtheorem{theorem}{Theorem}
\newtheorem{proposition}[theorem]{Proposition}
\theoremstyle{definition}
\newtheorem{definition}[theorem]{Definition}

\newcommand{\g}{\mathfrak{g}}

\usepackage[normalem]{ulem}

\newcommand{\p}{\partial}

\usepackage{amsmath,amssymb}

\begin{document}

\begin{abstract}
In this paper, we establish new area-charge inequalities for the boundary of time-symmetric Einstein-Maxwell initial data sets, in both compact and noncompact cases, under the dominant energy condition. These inequalities lead to novel rigidity theorems with no analogues in the uncharged setting. In the noncompact case, our result is obtained by applying Gromov's $\mu$-bubble technique in a new geometric context.
\end{abstract}

\maketitle

\section{Introduction}

Inequalities relating the physical quantities of black holes -- such as mass, charge, and angular momentum -- have been a central topic in mathematical relativity and geometric analysis. Numerous works have explored the geometric and physical constraints arising from the interplay among these quantities and the horizon area, particularly in the context of stable minimal surfaces, which are especially significant due to their variational characterization (see, e.g., \cite{Gibbons, Mendes2025}). Additional contributions include studies on isoperimetric surfaces \cite{DainJaramilloReiris}, stable marginally outer trapped surfaces (MOTS) \cite{Simon}, and min-max constructions leading to area-charge inequalities \cite{Cruz}. For a comprehensive overview of these and other geometric inequalities, we refer the reader to \cite{DainGabach-Clement}.

Using variational and minimal surface techniques, we establish sharp area-charge estimates for time-symmetric charged initial data sets with boundary, valid across compact and noncompact settings where the dominant energy condition holds.

Consider an oriented Riemannian 3-manifold $(M^3,g)$ and a vector field representing the electric field $E\in\mathfrak{X}(M)$. The triple $(M^3,g,E)$ can be viewed as part of a time-symmetric initial data set for the Einstein-Maxwell equations. In this context, the \emph{charged dominant energy condition} for the non-electromagnetic matter fields implies (see Section~\ref{models}):
$$
R_g \geq 2\Lambda + 2|E|^2,
$$
where $R_g$ is the scalar curvature of $(M^3,g)$ and $\Lambda$ is a constant representing the cosmological constant. In this framework, it is natural to associate with any closed orientable embedded surface $\Sigma\subset M$ the notion of enclosed \emph{charge}, defined by the flux integral
\begin{equation*}
Q(\Sigma)=\frac{1}{4\pi}\int_{\Sigma}\langle E,N\rangle, 
\end{equation*}
where $N$ is a unit normal vector field along $\Sigma$. Assuming that $\div_g E = 0$, it follows from the divergence theorem that the charge depends only on the homology class of $\Sigma$.

Our first result establishes geometric inequalities for charged Riemannian 3-manifolds with boundary. Given a boundary component $\Sigma$, the theorem relates the cosmological constant, the electric charge $Q(\Sigma)$, and the boundary geometry, providing constraints on the area of $\Sigma$ as well as on the global structure of $M$.

\begin{theorem}\label{compact_thm}
Let $(M^3,g,E)$ be a compact orientable Riemannian 3-manifold with boundary, equipped with a divergence-free vector field $E$ representing the electric field. Suppose the scalar curvature satisfies
\begin{equation*}
R_g \geq 2\Lambda + 2|E|^2,
\end{equation*}
and assume that the boundary $\p M$ is weakly mean-convex.

Let $\Sigma$ be a connected component of $\p M$ such that $H_2(M,\Sigma) = 0$.

\begin{enumerate}
    \item If $\Lambda > 0$, then 
    \begin{equation*}
    4\Lambda Q(\Sigma)^2 \le 1,
    \end{equation*}
    and
    \begin{equation}\label{equation.2}
    |\Sigma| \ge \frac{2\pi}{\Lambda} \Big(1 - \sqrt{1 - 4\Lambda Q(\Sigma)^2} \Big).
    \end{equation}
    \item If $\Lambda = 0$, then
    \begin{equation}\label{equation.3}
    |\Sigma| \ge 4\pi Q(\Sigma)^2.
    \end{equation}
\end{enumerate}

Moreover, if equality holds in \eqref{equation.2} and \eqref{equation.3}, then $(M^3,g)$ is isometric to the Riemannian product $([0, \ell] \times \Sigma, dt^2 + g_0)$ for some $\ell > 0$, where the induced metric $g_0$ on $\Sigma$ has constant Gaussian curvature $\kappa_g = a^2 + \Lambda$, and $E = aN$. 
In either case, the genus $\g(\Sigma)$ of $\Sigma$ equals zero, assuming the equality holds.
\end{theorem}

The proof of Theorem~\ref{compact_thm}, as well as the subsequent results, relies on a combination of techniques from Geometric Measure Theory and topological constraints. These tools enable the construction of area-minimizing surfaces, that is, surfaces that minimize area within a given class of competitors.  Among minimal surfaces, area-minimizing ones play a central role in various scalar curvature rigidity results, including those in \cite{BrayBrendleEichmairNeves, BrayBrendleNeves, CaiGalloway, ChodoshEichmairMoraru, MarquesNeves, MicallefMoraru, Nunes, Zhu2021}.

In general, when the charge vanishes, most existing results recover classical Riemannian conclusions. However, the results presented here have no counterpart  in the purely Riemannian setting, except for the case when $\Lambda<0$ and $\g(\Sigma)\ge2$ in the next theorem, which recovers \cite[Theorem~5]{Nunes}.

We now extend our analysis to the case of a negative cosmological constant $\Lambda < 0$. Before stating our next result, let us recall that a 3-manifold $M$ is \emph{irreducible} if every embedded 2-sphere in $M$ bounds an embedded 3-ball, and a surface $\Sigma \subset M$, other than a 2-sphere, is \emph{incompressible} if it is $\pi_1$-injective, that is, the induced map $\pi_1(\Sigma)\to\pi_1(M)$ is injective.

\begin{theorem}\label{isotopy_thm}
Let $(M^3,g,E)$ be a compact orientable Riemannian 3-manifold with boundary, equipped with a divergence-free vector field $E$ representing the electric field. Suppose the scalar curvature satisfies
\begin{equation*}
R_g \geq 2\Lambda + 2|E|^2,
\end{equation*}
for some constant $\Lambda<0$, and that the boundary $\p M$ is weakly mean-convex. 

Let $\Sigma$ be a connected component of $\p M$.
\begin{enumerate}
\item If $\inf |E|^2>|\Lambda|$ and $H_2(M,\Sigma)=0$, then
\begin{equation}\label{equation.4}
|\Sigma| \geq \frac{2\pi}{|\Lambda|} \Big( \sqrt{1 + 4|\Lambda| Q(\Sigma)^2}-1 \Big).
\end{equation}
\item If $\Sigma$ is incompressible, and $M$ is irreducible and does not contain any closed non-orientable embedded surfaces, then 
\begin{equation}\label{equation.5}
|\Sigma| \geq \frac{2\pi}{|\Lambda|} \Big( (\g(\Sigma) - 1) + \sqrt{(\g(\Sigma) - 1)^2 + 4|\Lambda| Q(\Sigma)^2} \Big),
\end{equation}
where $\g(\Sigma)$ is the genus of $\Sigma$. 
\end{enumerate}

Moreover, if equality holds in any of the above inequalities, then $(M^3,g)$ is isometric to the Riemannian product $([0, \ell] \times \Sigma, dt^2 + g_0)$ for some $\ell > 0$, where the induced metric $g_0$ on $\Sigma$ has constant Gaussian curvature $\kappa_g = a^2 + \Lambda$, and $E = aN$. 
\end{theorem}

It is interesting to note that Theorem~\ref{isotopy_thm} imposes no restriction on the genus of $\Sigma$. This contrasts with the uncharged case when $\Lambda < 0$, which only admits surfaces of genus at least two, and thus represents a new phenomenon. Moreover, in the case $\g(\Sigma)=1$, our result can be compared to \cite[Theorem~2]{CaiGalloway}, which asserts that if $M$ is a complete 3-manifold with nonnegative scalar curvature and weakly mean-convex boundary, then the existence of any two-sided torus $\Sigma\subset M$ that minimizes area in its isotopy class implies that $M$ is flat. In contrast, our result allows the presence of a 2-torus $\Sigma$ with an explicit lower bound for its area:
\begin{equation*}
|\Sigma| \ge \frac{4\pi |Q(\Sigma)|}{\sqrt{|\Lambda|}}.
\end{equation*}

For higher genus surfaces ($\g(\Sigma)\ge2)$, our result may also be compared to \cite[Theorem~5]{Nunes}.

The next result extends the previous theorems to the noncompact setting. The key tool in the proof is Gromov's concept of a $\mu$-bubble \cite{Gromov1996, Gromov2023}, which arises as a minimizer of a weighted area functional. More precisely, given a Riemannian manifold $(M^n,g)$ and a function $h$, a $\mu$-bubble is defined as a minimizer or, more generally, a critical point) of the functional
\begin{equation*}
\Omega\mapsto\mathcal{H}^{n-1}(\p\Omega)-\int_{\Omega}h,
\end{equation*}
where $\mathcal{H}^{n-1}$ is the induced $(n-1)$-dimensional Hausdorff measure, and $\Omega\subset M$ ranges over a suitable class of subsets, such as Caccioppoli sets.

Very recently, Gromov's $\mu$-bubble technique has emerged as a powerful tool in geometric analysis, with numerous applications to other geometric problems; see, for instance, \cite{CecchiniRadeZeidler, ChodoshLi2023, ChodoshLi2024, Gromov2020, LesourdUngerYau, Mazet, Wang, Zhu2021, Zhu2023, Zhu2024ArXiv, Zhu2024TAMS}.

\begin{theorem}\label{complete}
Let $(M^3,g,E)$ be a complete, noncompact, orientable Riemannian 3-manifold with connected compact boundary $\p M$, equipped with a divergence-free vector field $E$ representing the electric field. Assume that the scalar curvature satisfies
\begin{equation*}
R_g \geq 2\Lambda + 2|E|^2,
\end{equation*}
where $\Lambda$ is a constant such that $\Lambda + |E|^2$ is uniformly positive, and that the boundary $\p M$ is weakly mean-convex. Assume further that $H_2(M,\p M)=0$.  Then the following inequalities hold:
\begin{enumerate}
\item If $\Lambda\in\mathbb{R}\setminus\{0\}$, then
\begin{equation*}
4\Lambda Q(\p M)^2 \le 1
\end{equation*}
(this is trivial when $\Lambda<0$),
and
\begin{align}\label{iewq1}
|\p M| \ge \frac{2\pi}{\Lambda} \Big(1 - \sqrt{1 - 4\Lambda Q(\p M)^2} \Big).
\end{align}
\item If $\Lambda = 0$, then
\begin{align}\label{iewq2}
|\p M| \ge 4\pi Q(\p M)^2.
\end{align}
\end{enumerate}

Moreover, if equality holds in either \eqref{iewq1} or \eqref{iewq2}, then there exists an isometry
\begin{equation*}
\Phi: ([0,+\infty)\times\p M,dt^2+g_0)\to(M^3,g),
\end{equation*}
such that $\Phi(0, \p M) = \p M$, where the induced metric $g_0$ on $\p M$ has constant Gaussian curvature $\kappa_g = a^2 + \Lambda$, and $E = aN$.
\end{theorem}

The first difficulty in proving Theorem~\ref{complete} is to ensure that the sequence of approximating $\mu$-bubbles converges to a well-defined limiting surface, rather than diverging to infinity. This requires controlling the behavior of the $\mu$-bubbles to guarantee that they do not escape to infinity during the limiting process.

\medskip

\textbf{Outline of the paper.} In Section~\ref{models}, we introduce a family of model solutions that depends on the sign of the cosmological constant. These models serve both as motivating examples and as concrete illustrations of the general results established later in the paper. In Section~\ref{Area-charge estimates}, we derive sharp area-charge estimates under suitable geometric assumptions and prove an auxiliary local rigidity result, which plays an important role in the arguments developed in the subsequent section. Finally, in Section~\ref{proofs}, we present the proofs of Theorems~\ref{compact_thm}, \ref{isotopy_thm}, and \ref{complete}.

\medskip

\textbf{Acknowledgments:} T.~Cruz was partially supported by CNPq, Brazil, under grant numbers 307419/2022-3, 408834/2023-4, 444531/2024-6, and 403770/2024-6. A.~Mendes was partially supported by CNPq, Brazil, under grant number 309867/2023-1. The authors gratefully acknowledge partial support from Coordenação de Aperfeiçoamento de Pessoal de Nível Superior (CAPES/MATH-AMSUD 88887.985521/2024-00), Brazil.

\section{The models}\label{models}

We consider a time-symmetric initial data set (IDS) for the Einstein-Maxwell equations, denoted by $(M^3,g,E,B)$, where $E,B\in\mathfrak{X}(M)$ represent the electric and magnetic fields, respectively. The constraint equations in this setting are given by:
\begin{align*}
R_g - 2(|E|^{2} + |B|^{2}) - 2\Lambda &= 16\pi\mu, \\
\div_g E &= 4\pi\rho, \\
\div_g B &= 0,
\end{align*}
where $\mu$ denotes the energy density of non-electromagnetic matter, and $\rho$ is the electric charge density.  When $\rho \equiv 0$, we say that $(M^3,g,E,B)$ has no charged matter.

The dominant energy condition for the non-electromagnetic matter fields can be expressed in terms of the initial data set as:
\begin{align*}
R_g \geq 2(|E|^{2} + |B|^{2}) + 2\Lambda.
\end{align*}
 This terminology is justified by the well-known fact that the Cauchy development of such an IDS yields a Lorentzian 4-manifold solving the Einstein-Maxwell equations with the prescribed energy and charge densities.

For the sake of simplicity and clarity of exposition, we will set $B=0$ throughout this paper. However, it is not difficult to show that the results remain valid, with appropriate adaptations, in the present of magnetic fields.

We now present models that illustrate Theorems~\ref{compact_thm}, \ref{isotopy_thm}, and \ref{complete}.

\medskip

\noindent\textbf{Cosmological constant $\Lambda > 0$.} The first example we consider comes from the dS Bertotti-Robinson solution to the Einstein-Maxwell equations, given by
\begin{equation*}
\begin{cases}
ds^2 = \displaystyle\frac{1}{A}(-\sinh^2 \chi\, dT^2 + d\chi^2) + \frac{1}{B} d\Omega^2,\\
F = \displaystyle -Q \frac{B}{A} \sinh \chi\, dT \wedge d\chi,
\end{cases}
\end{equation*}
where $d\Omega^2$ denotes the round metric on $S^2$ of constant Gaussian curvature one. The parameters $A$ and $B$ are related to the cosmological constant $\Lambda$ and electric charge $Q$ through the identities
\begin{equation*}
\Lambda = \frac{B - A}{2} > 0, \quad Q^2 = \frac{A + B}{2B^2}.
\end{equation*}

Above, $F$ is the Faraday tensor associated with the Bertotti-Robinson solution. The corresponding electric vector field induced on a $T$-slice is given by
\begin{equation*}
E = QB \sqrt{A}\, \p_\chi = QBN,
\end{equation*}
where $N = \sqrt{A}\, \p_\chi$ is the unit normal in the $\chi$-direction. Taking $a =\langle E, N\rangle= QB$, direct computations show that the triple $([\chi_0, \chi_1] \times S^2, g, E)$, with
\begin{equation*}
g = \frac{1}{A} d\chi^2 + \frac{1}{B} d\Omega^2,
\end{equation*}
satisfies all the assumptions of Theorem~\ref{compact_thm} with $4\Lambda Q^2 < 1$ and equality in \eqref{equation.2}.

Our second example arises from the Nariai-Bertotti-Robinson solution:
\begin{align*}
\begin{cases}
ds^2 = -dT^2 + d\chi^2 + \dfrac{1}{2\Lambda} d\Omega^2,\\
F = -\sqrt{\Lambda}\, dT \wedge d\chi.
\end{cases}
\end{align*}
In this case, the electric vector field is given by $E = \sqrt{\Lambda}\, \p_{\chi}$. Taking $a = \sqrt{\Lambda}$, straightforward computations show that $([\chi_0,\chi_1] \times S^2, g, E)$, with
\begin{align*}
g = d\chi^2 + \frac{1}{2\Lambda} d\Omega^2,
\end{align*}
satisfies all the hypotheses of Theorem~\ref{compact_thm} with equality in \eqref{equation.2} and, in particular, in \eqref{equation.2}. The charge of each slice $\{\chi\} \times S^2$ is given by $Q = \frac{1}{2\sqrt{\Lambda}}$.

\medskip

\noindent\textbf{Cosmological constant $\Lambda = 0$.} Taking $A = B$ in the previous example, we obtain the so-called flat Bertotti-Robinson solution, which, as before, gives rise to an initial data set satisfying all the hypotheses of Theorem~\ref{compact_thm} with $\Lambda = 0$ and equality in \eqref{equation.3}.

\medskip

\noindent\textbf{Cosmological constant $\Lambda < 0$.} Our final example is a class of exact solutions that includes the AdS Bertotti-Robinson and the anti-Nariai spacetimes. These represent a natural generalization of the dS Bertotti-Robinson solution to the case where the cosmological constant is negative ($\Lambda < 0$):
\begin{align*}
\begin{cases}
ds^2 = \dfrac{1}{A}(-\sinh^2\chi\,dT^2 + d\chi^2) + \dfrac{1}{B}d\Omega_k^2, \\
F = -Q_{\rm ADM}\dfrac{B}{A}\sinh\chi\,dT \wedge d\chi,
\end{cases}
\end{align*}
where $d\Omega_k^2$ denotes a metric of constant Gaussian curvature $k = 1$, $0$, or $-1$ on a closed orientable surface $\Sigma$ of genus $\g(\Sigma) = 0$, $1$, or $\ge 2$ depending on whether $k = 1$, $0$, or $-1$, respectively.

The constants $A$ and $B$ are related to the cosmological constant $\Lambda$ and the ADM charge $Q_{\rm ADM}$ through the identities
\begin{align*}
\Lambda = -\frac{1}{2}(A - Bk) < 0, \quad Q_{\rm ADM}^2 = \frac{A + Bk}{2B^2}.
\end{align*}

As in the previous example, the associated electric vector field on a spacelike $T$-slice is given by
\begin{equation*}
E = Q_{\rm ADM}BN,
\end{equation*}
where $N = \sqrt{A}\,\p_\chi$.

Once more, letting $a = Q_{\rm ADM}B$, a direct computation shows that the triple $([\chi_0,\chi_1]\times\Sigma, g, E)$, where
\begin{align*}
g = \frac{1}{A}d\chi^2 + \frac{1}{B}d\Omega_k^2,
\end{align*}
satisfies the respective hypotheses of Theorem~\ref{isotopy_thm}, achieving equality in \eqref{equation.4} for $k=1$, or in \eqref{equation.5} for $k=0,-1$. Note that the electric charge of each slice $\{\chi\}\times\Sigma$ is given by $$Q=Q_{\rm ADM}\frac{\omega_k}{4\pi},$$ where $\omega_k$ is the area of $(\Sigma,d\Omega_k^2)$.

Finally, we point out that all the above metrics can be extended to complete smooth metrics on $([\chi_0, +\infty) \times \Sigma, g, E)$, thereby providing models for Theorem~\ref{complete}.

These and other special solutions to the Einstein-Maxwell equations can be found in \cite{CardosoDiasLemos}.

\section{Area-charge estimates and an auxiliary local rigidity result}\label{Area-charge estimates}

The goal of this section is to prove some area-charge estimates and to establish an auxiliary local rigidity result that will be instrumental in the proof of the main theorems.

\subsection{Area-charge estimates}

We begin by deriving area-charge estimates that relate the sign of the cosmological constant to the topology and the electric charge of a weakly outermost or area-minimizing surface.

\begin{proposition}\label{area estimate}
Let $(M^3,g)$ be a Riemannian 3-manifold with boundary, equipped with a divergence-free vector field $E$ representing the electric field. Suppose the scalar curvature of $(M^3,g)$ satisfies
\begin{align*}
R_g \ge 2\Lambda + 2|E|^2,
\end{align*}
for some constant $\Lambda$.

Assume further that a closed connected component $\Sigma$ of $\partial M$ is weakly mean-convex and either weakly outermost or area-minimizing. Then $\Sigma$ is minimal, and one of the following cases holds:

\begin{itemize}
\item[a)] \textbf{Case $\Lambda > 0$.} The surface $\Sigma$ is topologically a 2-sphere ($\g(\Sigma) = 0$), and its charge satisfies $4\Lambda Q(\Sigma)^2 \le 1$. In this case, the area of $\Sigma$ obeys
\begin{align}\label{eq1}
\frac{2\pi}{\Lambda}\Big(1 - \sqrt{1 - 4\Lambda Q(\Sigma)^2}\Big)
\le |\Sigma| \le
\frac{2\pi}{\Lambda}\Big(1 + \sqrt{1 - 4\Lambda Q(\Sigma)^2}\Big).
\end{align}

\item[b)] \textbf{Case $\Lambda = 0$.} We have $\g(\Sigma) = 0$ or $\g(\Sigma) = 1$.

\begin{itemize}
\item If $\g(\Sigma) = 0$, then
\begin{align}\label{eq2}
|\Sigma| \ge 4\pi Q(\Sigma)^2.
\end{align}

\item If $\g(\Sigma) = 1$, then necessarily $Q(\Sigma) = 0$.
\end{itemize}

\item[c)] \textbf{Case $\Lambda < 0$.} There is no restriction on the genus of $\Sigma$ ($\g(\Sigma) \ge 0$), and its area satisfies
\begin{align}\label{eq3}
|\Sigma| \ge \frac{2\pi}{|\Lambda|}\Big((\g(\Sigma) - 1) + \sqrt{(\g(\Sigma) - 1)^2 + 4|\Lambda| Q(\Sigma)^2}\Big).
\end{align}
\end{itemize}
Here, $Q(\Sigma)$ denotes the electric charge enclosed by $\Sigma$, and $\g(\Sigma)$ represents the genus of $\Sigma$.
\end{proposition}

\begin{proof}
Let $N$ be the inward-pointing unit normal vector to $\Sigma$, and let $H = \div_\Sigma N$ denote the mean curvature of $\Sigma$ with respect to $N$. The boundary component $\Sigma$ being weakly mean-convex means that its mean curvature vector $\vec{H} = -HN$ points into the manifold $M$, which is equivalent to $H \leq 0$.

We claim that $\Sigma$ is minimal. Otherwise, by evolving $\Sigma$ via mean curvature flow, we obtain a surface $\Sigma'$, close to $\Sigma$, whose mean curvature and area satisfy $H < 0$ and $|\Sigma'| < |\Sigma|$. This contradicts the assumption that $\Sigma$ is either weakly outermost or area-minimizing.

This shows that $\Sigma$ is minimal and either weakly outermost or area-minimizing; in particular, $\Sigma$ is a stable minimal surface. Therefore, for any $u \in C^\infty(\Sigma)$, we have
\begin{align*}
\frac{1}{2} \int_\Sigma (R_g + |A|^2) u^2 \leq \int_\Sigma |\nabla u|^2 + \int_\Sigma \kappa_g u^2,
\end{align*}
where $A$ denotes the second fundamental form of $\Sigma$, and $\kappa_g$ is the Gaussian curvature of $\Sigma$.

Taking $u \equiv 1$, using the inequality $R_g \geq 2\Lambda+2|E|^2$ and the Gauss-Bonnet theorem, we obtain
\begin{align*}
\Lambda|\Sigma|+\int_\Sigma |E|^2 \leq \int_\Sigma \kappa_g = 2\pi\chi(\Sigma),
\end{align*}
where $\chi(\Sigma)$ is the Euler characteristic of $\Sigma$. On the other hand, by the Cauchy-Schwarz inequality,
\begin{align*}
16 \pi^2 Q(\Sigma)^2 = \bigg(\int_\Sigma \langle E, N \rangle \bigg)^2 \leq|\Sigma| \int_\Sigma |E|^2.
\end{align*}
Therefore,
\begin{align}\label{eq:thm:local.aux1}
\Lambda |\Sigma| + \frac{16 \pi^2 Q(\Sigma)^2}{|\Sigma|} \leq 2\pi\chi(\Sigma).
\end{align}

It is straightforward to verify that inequality \eqref{eq:thm:local.aux1} leads to the area bounds stated in the proposition, depending on the sign of the cosmological constant $\Lambda$ and the topology of $\Sigma$ via the Euler characteristic (see, e.g., the proof of \cite[Proposition~3.1]{GallowayMendes2025}). We remark that, in the case $\Lambda > 0$, the proof of the inequality $4\Lambda Q(\Sigma)^2 \le 1$ can also be found in \cite{GallowayMendes2025}.
\end{proof}

\subsection{Auxiliary local rigidity result}

In this subsection, we establish the local rigidity result that follows from Proposition~\ref{area estimate}.

We begin by noting an infinitesimal rigidity statement: if equality holds in \eqref{eq1}, \eqref{eq2}, or \eqref{eq3}, then all inequalities in the proof of Proposition~\ref{area estimate} must, in fact, be equalities. This observation leads to the following:

\begin{proposition}\label{prop:infinitesimal}
 If $\Sigma$ attains the equality in \eqref{eq1}, \eqref{eq2}, and \eqref{eq3}, then the electric field satisfies $E = a N$ along $\Sigma$ for some constant $a$, $\Sigma$ is totally geodesic, $R_g = 2a^2 + 2\Lambda$ on $\Sigma$, and the function $u_0 \equiv 1$ is a Jacobi function on $\Sigma$, that is, 
\begin{align*}
 \Delta u_0 + \frac{1}{2}(R_g - 2\kappa_g + |A|^2)u_0 = 0.
 \end{align*}
 In particular, $\Sigma$ has constant Gaussian
curvature equal to $a^2 + \Lambda$ with respect to the induced metric.
\end{proposition}

We now establish the following local rigidity result, which will be essential in the proof of our main theorems. The arguments follow the same geometric construction employed in \cite{GallowayMendes2025}, \cite{LimaSousaBatista}, and \cite{Mendes2025} for analogous rigidity statements. We prove it here for the sake of completeness and convenience of the reader.

\begin{proposition}\label{rigidity}
If $\Sigma$ attains the equality in \eqref{eq1}, \eqref{eq2}, or \eqref{eq3}, then there exists a neighborhood $U \cong [0,\delta) \times \Sigma$ of $\Sigma$ in $M$ such that:
\begin{enumerate}
\item The electric field is normal to the foliation; more precisely, $E = a N_t$ for some constant $a$, where $N_t$ is the unit normal to $\Sigma_t \cong \{t\} \times \Sigma$ along the foliation.
\item $(U, g|_U)$ is isometric to $([0,\delta) \times \Sigma, dt^2 + g_0)$ for some $\delta > 0$, where the induced metric $g_0$ on $\Sigma$ has constant Gaussian curvature $\kappa_g = a^2 + \Lambda$.
\end{enumerate}
\end{proposition}

\begin{proof}
If equality holds in \eqref{eq1}, \eqref{eq2}, or \eqref{eq3}, then (in view of Proposition~\ref{prop:infinitesimal}) standard arguments (see, e.g., \cite{AnderssonCaiGalloway,BrayBrendleNeves,MicallefMoraru,Nunes}) imply that a collar neighborhood $U \cong [0, \delta) \times \Sigma$ of $\Sigma$ in $M$ can be foliated by constant mean curvature surfaces $\Sigma_t \cong \{t\} \times \Sigma$, with $\Sigma_0 = \Sigma$ and
\begin{align*}
g = \phi^2\,dt^2 + g_t \quad \text{on} \quad U,
\end{align*}
where $g_t$ is the induced metric on $\Sigma_t$. Let $H(t) = \div_{\Sigma_t} N_t$ denote the mean curvature of $\Sigma_t$ with respect to the unit normal vector $N_t = \phi^{-1} \partial_t$.

Recall the following well-known evolution equation for $H(t)$ under normal deformations of surfaces (see, for instance, \cite{HuiskenPolden}):
\begin{align*}
H'(t) = -\Delta_{\Sigma_t} \phi - \frac{1}{2}(R_g - 2\kappa_{\Sigma_t} + |A_{\Sigma_t}|^2 + H(t)^2)\phi.
\end{align*}
Dividing both sides by $\phi$ and integrating over $\Sigma_t$, we obtain
\begin{align*}
H'(t) \int_{\Sigma_t} \frac{1}{\phi}
&\le - \int_{\Sigma_t} \frac{|\nabla_{\Sigma_t} \phi|^2}{\phi^2}
- \frac{1}{2} \int_{\Sigma_t} (R_g + |A_{\Sigma_t}|^2 + H(t)^2)
+ \int_{\Sigma_t} \kappa_{\Sigma_t} \\
&\le - \int_{\Sigma_t} (|E|^2 + \Lambda) + 4\pi(1 - \g(\Sigma)),
\end{align*}
where we used the divergence theorem and the Gauss-Bonnet theorem, along with the inequality $R_g \ge 2|E|^2+2\Lambda$.

On the other hand, the Cauchy-Schwarz inequality implies
\begin{align*}
H'(t) \int_{\Sigma_t} \frac{1}{\phi}
\le -\frac{16\pi^2 Q(t)^2}{|\Sigma_t|} - \Lambda |\Sigma_t| + 4\pi(1 - \g(\Sigma)),
\end{align*}
where $Q(t)$ denotes the charge of $\Sigma_t$. Observe that, since $\div_g E = 0$, the divergence theorem yields
$$
Q(t) = Q(0) = Q(\Sigma), \quad \forall t \in [0, \delta).
$$
Therefore, under the assumption that equality holds in \eqref{eq:thm:local.aux1}, we obtain
\begin{align*}
H'(t) \int_{\Sigma_t} \frac{1}{\phi}
\le C \bigg( \frac{1}{|\Sigma|} - \frac{1}{|\Sigma_t|} \bigg)
+ \Lambda ( |\Sigma| - |\Sigma_t|),
\end{align*}
where $C = 16\pi^2 Q(\Sigma)^2$.

According to our assumptions, we have $|\Sigma| \le |\Sigma_t|$ for every $t \in [0, \delta)$. This is immediate when $\Sigma$ is area-minimizing. On the other hand, if $\Sigma$ is weakly outermost, we compute
\begin{align*}
|\Sigma| - |\Sigma_t| = -\int_0^t H(s) \bigg( \int_{\Sigma_s} \phi \bigg) ds \le 0,
\end{align*}
since $H(t) \ge 0$. Here, we applied the fundamental theorem of calculus together with the first variation of area formula. 
Therefore, if $\Lambda \ge 0$, we obtain
\begin{align*}
H'(t) \int_{\Sigma_t} \frac{1}{\phi}
&\le C \bigg( \frac{1}{|\Sigma|} - \frac{1}{|\Sigma_t|} \bigg)
= C \int_0^t \frac{H(s)}{|\Sigma_s|^2} \bigg( \int_{\Sigma_s} \phi \bigg) ds.
\end{align*}

When $\Lambda < 0$, we instead write
\begin{align*}
H'(t) \int_{\Sigma_t} \frac{1}{\phi}
&\le C \bigg( \frac{1}{|\Sigma|} - \frac{1}{|\Sigma_t|} \bigg) + \Lambda ( |\Sigma| - |\Sigma_t|)\\
&= \int_0^t \bigg( \frac{C}{|\Sigma_s|^2} - \Lambda \bigg) H(s) \bigg( \int_{\Sigma_s} \phi \bigg) ds.
\end{align*}

In either case, we conclude that
\begin{align*}
H'(t) \int_{\Sigma_t} \frac{1}{\phi} \le \int_0^t H(s)\xi(s)ds, \quad \forall t \in [0, \delta),
\end{align*}
for some function $\xi(t) \ge 0$. Thus, by Lemma~3.2 in \cite{Mendes2019}, it follows that $H(t) \le 0$ for all $t \in [0, \delta)$. If $\Sigma$ is weakly outermost, this guarantees that $H(t) = 0$ for all $t$, and all of the inequalities above must be equalities.

If $\Sigma$ is area-minimizing, we similarly find
\begin{align*}
0 \ge |\Sigma| - |\Sigma_t| = -\int_0^t H(s) \bigg( \int_{\Sigma_s} \phi \bigg) ds \ge 0,
\end{align*}
so again $H(t) = 0$ for all $t$, and all inequalities become equalities.

Standard computations as in \cite{Mendes2025} then yield the desired result.
\end{proof}

\section{Proof of the main results}\label{proofs}

This section is entirely devoted to the proof of Theorems~\ref{compact_thm}, \ref{isotopy_thm}, and \ref{complete}.

\subsection{Compact cases: Theorems~\ref{compact_thm} and \ref{isotopy_thm}}

\begin{proof}[Proof of Theorem~\ref{compact_thm}]
We divide the proof into two cases:

\medskip
\noindent
\textbf{Case $\Lambda > 0$.} We first show that $4\Lambda Q(\Sigma)^2 \leq 1$. 

Assume, by contradiction, that $4\Lambda Q(\Sigma)^2 > 1$. Since $[\Sigma] \neq 0$ in $H_2(M;\mathbb{Z})$ (as, otherwise, we would have $\Sigma = \partial M$ and, in particular, $Q(\Sigma) = 0$), we may apply existence results from Geometric Measure Theory \cite{Federer} for area-minimizing surfaces in integral homology classes. This yields a smooth, embedded, closed, oriented, area-minimizing surface $\Sigma'$ that is homologous to $\Sigma$. Furthermore, since $H_2(M, \Sigma) = 0$ and $\Sigma$ is connected, $\Sigma'$ must itself be connected.

If $\Sigma' \cap \partial M \neq \varnothing$, then by the weak mean-convexity of $\partial M$, the surface $\Sigma'$ must coincide with a connected component of $\partial M$. Proposition~\ref{area estimate} then implies that $\Sigma'$ is minimal and satisfies
\begin{align*}
4\Lambda Q(\Sigma)^2 > 1 \geq 4\Lambda Q(\Sigma')^2.
\end{align*}
However, since $\Sigma$ and $\Sigma'$ are homologous, they enclose the same total charge; that is, $Q(\Sigma) = Q(\Sigma')$. This yields a contradiction.

If instead $\Sigma' \cap \partial M = \varnothing$, then $\Sigma'$ is an area-minimizing minimal surface entirely contained in the interior of $M$. In this case, the same inequality holds:
\begin{align*}
4\Lambda Q(\Sigma)^2 > 1 \geq 4\Lambda Q(\Sigma')^2,
\end{align*}
again leading to a contradiction.

Thus, in either case, we conclude that $4\Lambda Q(\Sigma)^2 \leq 1$.

Next, assume
\begin{align*}
|\Sigma| \leq \frac{2\pi}{\Lambda} \Big(1 - \sqrt{1 - 4\Lambda Q(\Sigma)^2} \Big).
\end{align*}

As above, this inequality implies that $[\Sigma] \neq 0$ in $H_2(M;\mathbb{Z})$. Again, let $\Sigma'$ be a connected area-minimizing surface in the homology class of $\Sigma$. As before, whether or not $\Sigma'$ intersects $\partial M$, it is an area-minimizing minimal surface and, by Proposition~\ref{area estimate},
\begin{align*}
\frac{2\pi}{\Lambda} \Big(1 - \sqrt{1 - 4\Lambda Q(\Sigma)^2} \Big) 
&\geq |\Sigma| 
\geq |\Sigma'| 
\geq \frac{2\pi}{\Lambda} \Big(1 - \sqrt{1 - 4\Lambda Q(\Sigma)^2} \Big),
\end{align*}
from which it follows that
\begin{align*}
|\Sigma| = |\Sigma'| = \frac{2\pi}{\Lambda} \Big(1 - \sqrt{1 - 4\Lambda Q(\Sigma)^2} \Big).
\end{align*}
In particular, $\Sigma$ is also area-minimizing. We used here that $Q(\Sigma) = Q(\Sigma')$.

Since $\Sigma$ is area-minimizing and its area attains the first equality in \eqref{eq1}, we may apply Proposition~\ref{area estimate} to see that $\mathfrak{g}(\Sigma) = 0$ and Proposition~\ref{rigidity} to obtain a collar neighborhood $U \cong [0, \delta) \times \Sigma$ of $\Sigma$ in $M$ satisfying the conclusion of the theorem.

Let $\Sigma_t \cong \{t\} \times \Sigma$. As $t \nearrow \delta$, by standard compactness results for stable minimal surfaces with uniformly bounded area, the surfaces $\Sigma_t$ converge (subsequentially) to a closed, embedded, area-minimizing surface $\Sigma_\delta$. Since $(M^3, g)$ has positive scalar curvature, it follows that $\Sigma_\delta$ is diffeomorphic to either $\mathbb{S}^2$ or $\mathbb{RP}^2$. In the latter case, $M^3$ is diffeomorphic to $\mathbb{RP}^3 \setminus \text{ball}$, which contradicts the fact that $M^3$ has more than one boundary component\linebreak (as, otherwise, we would have $Q(\Sigma)=0$).

Let $S:=\partial M\setminus\Sigma\neq\varnothing$. If $\Sigma_\delta \cap S \neq \varnothing$, then the maximum principle implies $\Sigma_\delta = S$, establishing the result. Otherwise, if $\Sigma_\delta \cap S = \varnothing$, we may replace $\Sigma$ by $\Sigma_\delta$ and reapply Proposition~\ref{rigidity}. This replacement is valid because $|\Sigma_\delta| = |\Sigma|$ and $Q(\Sigma_\delta) = Q(\Sigma)$.  Let $T>0$ the maximal time in which this flow exists and is smooth. The conclusion then follows by a continuity argument, extending the local splitting to the entire manifold $M$.

\medskip
\noindent
\textbf{Case $\Lambda = 0$.} 
Assume that
\begin{align*}
|\Sigma| \leq 4\pi Q(\Sigma)^2.
\end{align*}
Let $\Sigma'$ be as before. It follows from Proposition~\ref{area estimate} that
\begin{align*}
4\pi Q(\Sigma)^2 \geq |\Sigma| \geq |\Sigma'| \geq 4\pi Q(\Sigma)^2,
\end{align*}
and hence $|\Sigma| = |\Sigma'| = 4\pi Q(\Sigma)^2$. In particular, $\Sigma$ is area-minimizing and attains the lower bound in \eqref{eq2}. Thus, as in the case $\Lambda > 0$, we may apply Proposition~\ref{area estimate} to classify the topology of $\Sigma$ and Proposition~\ref{rigidity} to obtain the desired local splitting. The global extension then follows from the same continuity argument.
\end{proof}

A fundamental tool in the proof of our next theorem is the classical result of Meeks-Simon-Yau \cite{MeeksSimonYau}, which asserts the existence of an area-minimizing surface in the isotopy class of a connected, incompressible surface $\Sigma \subset M$, assuming that $M$ is an irreducible , closed  Riemannian 3-manifold. This result remains valid when $\partial M \neq \emptyset$, provided that the mean curvature of $\partial M$ with respect to the outward-pointing normal vector is nonnegative, that is, $\p M$ is weakly mean-convex (see Section~6 of \cite{MeeksSimonYau}).
We also point out that Hass and Scott \cite{HassScott} gave an alternative proof of this result without resorting to Geometric Measure Theory (cf.\ Theorem~5.1 in \cite{HassScott}).

We are now ready to prove the second main theorem.

\begin{proof}[Proof of Theorem~\ref{isotopy_thm}]

We divide the proof into two cases:

\medskip
\noindent\textbf{Case 1: $\inf |E|^2 > |\Lambda|$ and $H_2(M, \Sigma) = 0$.}

Assume that
\begin{align*}
|\Sigma| \le \frac{2\pi}{|\Lambda|} \Big( \sqrt{1 + 4|\Lambda| Q(\Sigma)^2} - 1 \Big).
\end{align*}

As in the proof of Theorem~\ref{compact_thm}, let $\Sigma'\subset M$ be a closed, connected area-minimizing minimal surface in the homology class of $\Sigma$. Note that $\Sigma'$ is topologically a 2-sphere, since $R_g \ge 2\Lambda + 2|E|^2 > 0$.

It follows from Proposition~\ref{area estimate} that
\begin{align*}
\frac{2\pi}{|\Lambda|} \Big( \sqrt{1 + 4|\Lambda| Q(\Sigma)^2} - 1 \Big)\ge|\Sigma|\ge|\Sigma'|\ge\frac{2\pi}{|\Lambda|}\Big( \sqrt{1 + 4|\Lambda| Q(\Sigma)^2} - 1 \Big),
\end{align*}
that is,
\begin{align*}
|\Sigma|=|\Sigma'|=\frac{2\pi}{|\Lambda|}\Big( \sqrt{1 + 4|\Lambda| Q(\Sigma)^2} - 1 \Big).
\end{align*}
The argument then proceeds similarly to the proof of Theorem~\ref{compact_thm}.

\medskip
\noindent\textbf{Case 2: $\Sigma$ is incompressible, and $M$ is irreducible and does not contain any closed non-orientable embedded surfaces.}

Denote by $\mathcal{J}(\Sigma)$ the isotopy class of $\Sigma$ in $M$. Under our assumptions, we can apply the version of Theorem~5.1 in \cite{HassScott} (see also \cite{MeeksSimonYau}) for 3-manifolds with nonempty boundary to obtain a closed embedded surface $\Sigma' \in \mathcal{J}(\Sigma)$ such that
$$
|\Sigma'| = \inf_{\hat{\Sigma} \in \mathcal{J}(\Sigma)} |\hat{\Sigma}|.
$$
Moreover, since $\Sigma$ and $\Sigma'$ are diffeomorphic and homologous, we have
\begin{align}\label{conseq}
\g(\Sigma') = \g(\Sigma) \quad \text{and} \quad Q(\Sigma') = Q(\Sigma).
\end{align}

On the other hand, since the boundary $\partial M$ is weakly mean-convex, the maximum principle implies that either $\Sigma'$ is contained in the interior of $M$ or it coincides with one of the connected components of $\partial M$. In any case, it follows from Proposition~\ref{area estimate} that
\begin{align*}
|\Sigma'| \ge \frac{2\pi}{|\Lambda|} \Big( (\g(\Sigma') - 1) + \sqrt{(\g(\Sigma') - 1)^2 + 4|\Lambda| Q(\Sigma')^2} \Big).
\end{align*}

Using \eqref{conseq} together with the fact that $\Sigma'$ is area-minimizing, we conclude that
\begin{align*}
|\Sigma| \ge \frac{2\pi}{|\Lambda|} \Big( (\g(\Sigma) - 1) + \sqrt{(\g(\Sigma) - 1)^2 + 4|\Lambda| Q(\Sigma)^2} \Big).
\end{align*}
Furthermore, if equality holds, then
\begin{align*}
|\Sigma| = |\Sigma'| = \frac{2\pi}{|\Lambda|} \Big( (\g(\Sigma) - 1) + \sqrt{(\g(\Sigma) - 1)^2 + 4|\Lambda| Q(\Sigma)^2} \Big).
\end{align*}

The desired conclusion then follows similarly to the rigidity statement in Theorem~\ref{compact_thm} (see also the proof of Theorem~5 in \cite{Nunes}).
\end{proof}

\subsection{Noncompact case: Theorem~\ref{complete}}

Before proving the next result, we recall some facts about $\mu$-bubbles in our setting. 

Let $(M^3, g)$ be a complete, noncompact, orientable Riemannian 3-manifold with connected compact boundary $\partial M$. By an argument similar to the proof of Lemma~2.1 in \cite{Zhu2023} (see also the proof of Proposition~1.3 in \cite{Zhu2024ArXiv}), there exists a proper, surjective, smooth function $\phi : M \to [0, +\infty)$ satisfying
\begin{align*}
\partial M = \phi^{-1}(0) \quad \text{and} \quad \operatorname{Lip} \phi < 1.
\end{align*}

Given a smooth function $h:[0,T)\to\mathbb{R}$, we define the functional 
\begin{align*}
\mu^h(\Omega)=\mathcal{H}^2(\partial^*\Omega)-\int_{M\setminus\Omega}h\circ\phi\,\mathrm{d}\mathcal{H}^3
\end{align*}
on the class of Caccioppoli sets $\Omega\subset M$ with reduced boundary $\p^*\Omega$, subject to the condition
\begin{align*}
M\setminus\Omega\Subset\phi^{-1}([0,T)).
\end{align*}
For a reference on Caccioppoli sets, see \cite{Giusti}.

Let $\Omega(t)$ be a smooth one-parameter family of regions with $\Omega(0) = \Omega$ and normal speed $\varphi$ at $t = 0$. The \emph{first variation} of $\mu^h$ is given by
\begin{equation*}
\frac{d}{dt} \bigg|_{t=0} \mu^h(\Omega(t)) = \int_{\partial \Omega} (H - h \circ \phi) \varphi,
\end{equation*}
where $H$ is the mean curvature of $\partial \Omega$ with respect to the outward-pointing unit normal vector field $N$ along $\partial \Omega$. 

Since $\varphi$ is arbitrary, if $\Omega$ is \emph{critical} for $\mu^h$, we conclude that $H = h \circ \phi$ along $\partial \Omega$. In this case, we say that $\Omega$ is a \emph{$\mu$-bubble}, and $\p\Omega$ is a surface with prescribed mean curvature $h \circ \phi|_{\p\Omega}$.

Given a $\mu$-bubble $\Omega(0)=\Omega$, the \emph{second variation formula} for $\mu^h$ is given by  
\begin{equation}\label{second}
\frac{d^2}{dt^2}\bigg|_{t=0} \mu^h(\Omega(t)) = \int_{\partial \Omega} |\nabla \varphi|^2 - ( \operatorname{Ric}_g(N, N)+|A|^2 + \langle \nabla (h \circ \phi), N \rangle ) \varphi^2,
\end{equation}
where $\operatorname{Ric}_g$ is the Ricci tensor of $(M^3,g)$, and $A$ is the second fundamental form of $\p\Omega\subset(M^3,g)$.

We say that a $\mu$-bubble $\Omega$ is \emph{stable} if the second variation of $\mu^h$ is nonnegative for all normal variations $\Omega(t)$ of $\Omega$, i.e., if the right-hand side of \eqref{second} is nonnegative for all smooth functions $\varphi$.

Recall that a \emph{band} is a connected compact manifold $\hat{M}$ together with a decomposition of its boundary:
\begin{equation*}
\partial \hat{M} = \partial_{-} \,\dot{\cup}\, \partial_{+},
\end{equation*}
where $\partial_{-}$ and $\partial_{+}$ are nonempty unions of connected components of $\partial \hat{M}$.

\begin{definition}
Let $(\hat{M}, \hat{g}, \partial_{-}, \partial_{+})$ be a Riemannian band. A smooth function $b : \operatorname{int}(\hat{M}) \to \mathbb{R}$ is said to satisfy the \emph{barrier condition} if, for each connected component $S \subset \partial_{+}$ (resp., $S \subset \partial_{-}$), either:
\begin{itemize}
\item $b$ smoothly extends to $S$ and satisfies $H_S\ge b|_S$ (resp., $H_S\ge-b|_S$), where $H_S$ is the mean curvature of $S$ with respect to the outward normal; or
\item $b\to-\infty$ (resp., $b\to+\infty$) towards $S$.
\end{itemize}
\end{definition}

The existence and regularity of a minimizer for $\mu^h$ among Caccioppoli sets, assuming the barrier condition (and in fact for a more general functional), were initially claimed by Gromov (see Section~5.1 of \cite{Gromov2023}). A rigorous proof and detailed exposition were later provided by Zhu \cite[Proposition~2.1]{Zhu2021} and by Chodosh and Li \cite[Proposition~12]{ChodoshLi2024}.

\begin{proof}[Proof of Theorem~\ref{complete}]
Let $\phi:M\to[0,+\infty)$ be as above.

According to \cite[Lemma~2.3]{Zhu2023} (see also the proof of Proposition~1.3 in \cite{Zhu2024ArXiv}), for any $\varepsilon\in(0,1)$, we can construct a smooth function $h_\varepsilon:[0,\frac{1}{3\varepsilon})\to(-\infty,0]$ such that:
\begin{itemize}
\item $h_\varepsilon$ satisfies
\begin{align*}
\frac{3}{2}h_\varepsilon^2+2h_\varepsilon'=6\varepsilon^2\quad\mbox{on}\quad\big[\tfrac{1}{6},\tfrac{1}{3\varepsilon}\big)
\end{align*}
and
\begin{align*}
\bigg|\frac{3}{2}h_\varepsilon^2+2h_\varepsilon'\bigg|\le C\varepsilon,
\end{align*}
where $C>0$ is a universal constant;
\item $h_\varepsilon(0)=0$, $h_\varepsilon'<0$, and
\begin{align*}
\lim_{t\to\frac{1}{3\varepsilon}}h_\varepsilon(t)=-\infty;
\end{align*}
\item as $\varepsilon\to0$, $h_\varepsilon$ converges smoothly to the zero function on any closed interval contained in $[0,\frac{1}{3\varepsilon})$.
\end{itemize}

Let $\varepsilon\in(0,1)$ be such that $\frac{1}{3 \varepsilon}$ is a regular value of $\phi$, and define the following prescribed-mean-curvature functional:
\begin{equation*}
\mu^\varepsilon(\Omega)=\mathcal{H}^{2}(\partial^*\Omega)-\int_{M \setminus\Omega}h_{\varepsilon}\circ\phi\,\mathrm{d}\mathcal{H}^{3}
\end{equation*}
on the class $\mathcal{C}_\varepsilon$ of Caccioppoli sets in $M$, with reduced boundary $\partial^*\Omega$, such that
\begin{equation*}
M\setminus\Omega\Subset\phi^{-1}\big(\big[0,\tfrac{1}{3\varepsilon}\big]\big).
\end{equation*}

For each $\varepsilon$, define the Riemannian band
\begin{equation*}
M_{\varepsilon} = \phi^{-1}\big(\big[0, \tfrac{1}{3\varepsilon}\big]\big),
\end{equation*}
with boundary decomposition
\begin{align*}
\partial_- = \phi^{-1}(0)=\p M,\quad \partial_+ &= \phi^{-1}\big(\tfrac{1}{3\varepsilon}\big).
\end{align*}

Note that each $h_{\varepsilon}\circ\phi$, restricted to $\operatorname{int}(M_\varepsilon)$, satisfies the barrier condition (since $\p M$ is weakly mean-convex). Combined with Sard's theorem, this implies that for each $\varepsilon \in (0,1)$, there exists a smooth minimizer $\Omega_\varepsilon \in \mathcal{C}_\varepsilon$ for the functional $\mu^\varepsilon$; see \cite[Proposition~2.1]{Zhu2021} and \cite[Proposition~12]{ChodoshLi2024}.

By taking a sequence of positive numbers $\varepsilon_k \rightarrow 0$ as $k \rightarrow \infty$, we thus obtain a corresponding sequence of smooth minimizers $\Omega_k$ in $\mathcal{C}_{\varepsilon_k}$ for the functionals $\mu^{k} := \mu^{\varepsilon_k}$.

Since each boundary $\partial\Omega_k$ is homologous to the surface $\partial M$, we may select the homologically nontrivial components of $\partial\Omega_k$, whose union we denote by $\Sigma_k$. On the other hand, because $\p M$ is connected, $\Sigma_k$ is homologous to $\p M$, and $H_2(M,\p M)=0$, we obtain that $\Sigma_k$ is also connected. Then, it follows from the second variation formula and the Gauss equation that,
\begin{eqnarray*}
\int_{\Sigma_k}|\nabla\varphi|^2&\ge&\int_{\Sigma_k}(\operatorname{Ric}_g(N_k,N_k)+|A_k|^2+\langle\nabla(h_{\varepsilon_k}\circ\phi),N_k\rangle)\varphi^2\\
&=&\frac{1}{2}\int_{\Sigma_k}(R_g-R_{\Sigma_k}+|A_k|^2+H_k^2+2\langle\nabla(h_{\varepsilon_k}\circ\phi),N_k\rangle)\varphi^2\\
&\ge&\frac{1}{2}\int_{\Sigma_k}\Big(\Big|A_k-\frac{H_k}{2}g_k\Big|^2-R_{\Sigma_k}\Big)\varphi^2\\
&&+\frac{1}{2}\int_{\Sigma_k}\Big(R_g+\Big(\frac{3}{2}h_{\varepsilon_k}^2+2h_{\varepsilon_k}'\Big)\circ\phi\Big)\varphi^2,
\end{eqnarray*}
for each $\varphi\in C^\infty(\Sigma_k)$, where $N_k$ is the outward unit normal to $\Sigma_k$, $A_k$ is the second fundamental form of $\Sigma_k$, $H_k$ is the mean curvature of $\Sigma_k$ with respect to $N_k$, and $R_{\Sigma_k}$ is the scalar curvature of $\Sigma_k$ with respect to the induced metric $g_k$. Above we used that 
\begin{align*}
\langle\nabla(h_{\varepsilon_k}\circ\phi),N_k\rangle\ge-|h_{\varepsilon_k}'\circ\phi|=h_{\varepsilon_k}'\circ\phi,
\end{align*}
since $h_{\varepsilon_k}'<0$ and $\operatorname{Lip}\phi<1$.

By taking the test function $\varphi \equiv 1$, the fact that $R_g \geq 2\Lambda + 2|E|^2$, and the Gauss-Bonnet theorem imply that
\begin{eqnarray*}
2\pi\chi(\Sigma_k)&\ge &\frac{1}{2}\int_{\Sigma_k}(R_g-C\varepsilon_k)\\
&\ge &\int_{\Sigma_k}(\Lambda+|E|^2)-\frac{1}{2}C\varepsilon_k|\Sigma_k|.
\end{eqnarray*}

Since we are assuming the function $\Lambda + |E|^2$ is uniformly positive, say $\Lambda + |E|^2 \geq \delta > 0$, we have
\begin{align*}
2\pi\chi(\Sigma_k) \geq \Big( \delta - \frac{1}{2} C\varepsilon_k \Big) |\Sigma_k|.
\end{align*}
Therefore, $\chi(\Sigma_k) > 0$, that is, $\Sigma_k$ is topologically a 2-sphere for $k$ sufficiently large. In this case,
\begin{eqnarray*}
4\pi &\ge& \int_{\Sigma_k}(\Lambda+|E|^2)-\frac{1}{2}C\varepsilon_k|\Sigma_k|\\
&\ge &\Lambda |\Sigma_k|+\frac{16\pi^2 Q(\Sigma_k)^2}{|\Sigma_k|} - \frac{1}{2} C\varepsilon_k |\Sigma_k|.
\end{eqnarray*}
Then we obtain
\begin{align*}
4\pi|\Sigma_k|\ge\Big(\Lambda-\frac{1}{2}C\varepsilon_k\Big)|\Sigma_k|^2+16\pi^2Q(\Sigma_k)^2.
\end{align*}

We now analyze two cases based on the cosmological constant $\Lambda$:

\medskip
\noindent
\textbf{Case $\Lambda\in\mathbb{R}\setminus\{0\}$.} Taking $k$ sufficiently large so that $\Lambda-\frac{1}{2}C\varepsilon_k>0$ if $\Lambda>0$, we have

\begin{align*}
4\Big(\Lambda-\frac{1}{2}C\varepsilon_k\Big)Q(\Sigma_k)^2\le 1
\end{align*}
and
\begin{align*}
|\Sigma_k|\ge\frac{2\pi}{\big(\Lambda-\tfrac{1}{2}C\varepsilon_k\big)}\bigg(1-\sqrt{1-4\big(\Lambda-\tfrac{1}{2}C\varepsilon_k\big)Q(\Sigma_k)^2}\bigg).
\end{align*}

\medskip
\noindent
\textbf{Case $\Lambda=0$.} In this case,
\begin{align*}
|\Sigma_k|\ge4\pi Q(\Sigma_k)^2-\frac{1}{8\pi}C\varepsilon_k|\Sigma_k|^2.
\end{align*}

In either case, observe that
\begin{align*}
|\Sigma_k| \leq |\partial \Omega_k| \leq \mu^k(\Omega_k) \leq \mu^k(M) = |\partial M|.
\end{align*}
Here we used that $h_{\varepsilon_k} \leq 0$ and that $\Omega_k$ minimizes $\mu^k$. Therefore, by taking the limit as $k \to +\infty$, and using the fact that $Q(\Sigma_k) = Q(\partial M)$ for each $k$, we obtain inequalities \eqref{iewq1} and \eqref{iewq2}.

Assume the equality holds in \eqref{iewq1} (resp., \eqref{iewq2}). We claim that $\Sigma_k$ must have nonempty intersection with the compact subset $\mathcal{K} = \phi^{-1}([0,\frac{1}{6}])$. In fact, if $\mathcal{K} \cap \Sigma_k = \varnothing$, then $\Sigma_k \subset \phi^{-1}((\frac{1}{6},\frac{1}{3\varepsilon_k}))$, and repeating the earlier argument would yield
\begin{align*}
4\pi |\Sigma_k| > \Lambda |\Sigma_k|^2 + 16\pi^2 Q(\Sigma_k)^2,
\end{align*}
since
\begin{align*}
\frac{3}{2} h_{\varepsilon_k}^2 + 2 h_{\varepsilon_k}' = 6 \varepsilon_k^2 \quad \text{on} \quad \big[\tfrac{1}{6},\tfrac{1}{3\varepsilon_k}\big),
\end{align*}
which leads to a contradiction provided \eqref{iewq1} (resp., \eqref{iewq2}) would be strict.

Now, since the surfaces $\Sigma_k$ are $h_{\varepsilon_k}$-minimizing boundaries with uniformly bounded area, we can invoke the curvature estimates in \cite[Theorem~3.6]{XinZhu}. Thus, after passing to a subsequence if necessary, $\Sigma_k$ converges smoothly, in a locally graphical sense with multiplicity one, to an area-minimizing minimal boundary $\Sigma$.

Note that $\Sigma$ may be either compact or noncompact. To rule out the noncompact case, we proceed as follows. Let $R_\Sigma$ and $A_\Sigma$ denote the scalar curvature and second fundamental form of $\Sigma$, respectively. Since $\Lambda + |E|^2$ is uniformly positive, and the stability of $\Sigma$ implies the nonnegativity of the operator
$$
-\Delta_\Sigma + \frac{1}{2}R_\Sigma - \frac{1}{2}(|A_\Sigma|^2 + \Lambda + |E|^2),
$$
it follows that $\Sigma$ must consist of spherical components (see \cite[Lemma~4.1]{Zhu2023} or \cite[Theorem~8.8]{GromovLawson}). Moreover, since each $\Sigma_k$ is connected, the limit $\Sigma$ must be a 2-sphere, and for sufficiently large $k$, each $\Sigma_k$ becomes a graph over $\Sigma$. In particular, $\Sigma$ is homologous to $\Sigma_k$.

Using the inequality $|\partial M| \geq |\Sigma|$ together with the stability inequality for minimal surfaces, we deduce that, in the case $\Lambda \in \mathbb{R} \setminus \{0\}$, the following chain of inequalities holds:
$$
\frac{2\pi}{\Lambda}\Big(1 - \sqrt{1 - 4\Lambda Q(\partial M)^2}\Big)
= |\partial M| \geq |\Sigma| 
\geq \frac{2\pi}{\Lambda}\Big(1 - \sqrt{1 - 4\Lambda Q(\partial M)^2}\Big),
$$
from which it follows that
$$
|\partial M| = |\Sigma| = \frac{2\pi}{\Lambda}\Big(1 - \sqrt{1 - 4\Lambda Q(\partial M)^2}\Big).
$$
Above we used that $Q(\Sigma)=Q(\Sigma_k)=Q(\p M)$.

This shows that $\partial M$ is an area-minimizing surface that saturates either the first inequality in \eqref{eq1} or the inequality \eqref{eq3}, depending on the sign of~$\Lambda$ (the analysis in the case $\Lambda = 0$ is analogous). Therefore, by Proposition~\ref{rigidity}, there exists a collar neighborhood $U \cong [0, \rho) \times \partial M$ of $\partial M$ in $M$ satisfying the conclusions of the theorem. The result then follows by a continuity argument.
\end{proof}

\bibliographystyle{amsplain}
\bibliography{bibliography3.bib}

\providecommand{\bysame}{\leavevmode\hbox to3em{\hrulefill}\thinspace}
\providecommand{\MR}{\relax\ifhmode\unskip\space\fi MR }
\providecommand{\MRhref}[2]{%
  \href{http://www.ams.org/mathscinet-getitem?mr=#1}{#2}
}
\providecommand{\href}[2]{#2}
\begin{thebibliography}{10}

\bibitem{AnderssonCaiGalloway}
Lars Andersson, Mingliang Cai, and Gregory~J. Galloway, \emph{Rigidity and
  positivity of mass for asymptotically hyperbolic manifolds}, Ann. Henri
  Poincar{\'e} \textbf{9} (2008), no.~1, 1--33 (English).

\bibitem{BrayBrendleEichmairNeves}
H.~Bray, S.~Brendle, M.~Eichmair, and A.~Neves, \emph{Area-minimizing
  projective planes in 3-manifolds}, Commun. Pure Appl. Math. \textbf{63}
  (2010), no.~9, 1237--1247 (English).

\bibitem{BrayBrendleNeves}
Hubert Bray, Simon Brendle, and Andre Neves, \emph{Rigidity of area-minimizing
  two-spheres in three-manifolds}, Commun. Anal. Geom. \textbf{18} (2010),
  no.~4, 821--830 (English).

\bibitem{CaiGalloway}
Mingliang Cai and Gregory~J. Galloway, \emph{Rigidity of area minimizing tori
  in 3-manifolds of nonnegative scalar curvature}, Commun. Anal. Geom.
  \textbf{8} (2000), no.~3, 565--573 (English).

\bibitem{CardosoDiasLemos}
Vitor Cardoso, \'Oscar J.~C. Dias, and Jos\'e P.~S. Lemos, \emph{{N}ariai,
  {B}ertotti-{R}obinson, and anti-{N}ariai solutions in higher dimensions},
  Phys. Rev. D \textbf{70} (2004), 024002.

\bibitem{CecchiniRadeZeidler}
Simone Cecchini, Daniel R{\"a}de, and Rudolf Zeidler, \emph{Nonnegative scalar
  curvature on manifolds with at least two ends}, J. Topol. \textbf{16} (2023),
  no.~3, 855--876 (English).

\bibitem{ChodoshEichmairMoraru}
Otis Chodosh, Michael Eichmair, and Vlad Moraru, \emph{A splitting theorem for
  scalar curvature}, Commun. Pure Appl. Math. \textbf{72} (2019), no.~6,
  1231--1242 (English).

\bibitem{ChodoshLi2023}
Otis Chodosh and Chao Li, \emph{Stable anisotropic minimal hypersurfaces in
  \(\mathbf{R}^4\)}, Forum Math. Pi \textbf{11} (2023), 22 (English), Id/No e3.

\bibitem{ChodoshLi2024}
\bysame, \emph{Generalized soap bubbles and the topology of manifolds with
  positive scalar curvature}, Ann. Math. (2) \textbf{199} (2024), no.~2,
  707--740 (English).

\bibitem{Cruz}
Tiarlos Cruz, Vanderson Lima, and Alexandre de~Sousa, \emph{Min-max minimal
  surfaces, horizons and electrostatic systems}, J. Differ. Geom. \textbf{128}
  (2024), no.~2, 583--637 (English).

\bibitem{DainGabach-Clement}
Sergio Dain and Mar{\'\i}a~Eugenia Gabach-Clement, \emph{Geometrical
  inequalities bounding angular momentum and charges in {G}eneral
  {R}elativity}, Living Reviews in Relativity \textbf{21} (2018), no.~5, 1--74.

\bibitem{DainJaramilloReiris}
Sergio Dain, Jos{\'e}~Luis Jaramillo, and Mart{\'{\i}}n Reiris,
  \emph{Area-charge inequality for black holes}, Classical Quantum Gravity
  \textbf{29} (2012), no.~3, 15 (English), Id/No 035013.

\bibitem{Federer}
Herbert Federer, \emph{Geometric measure theory}, Grundlehren Math. Wiss., vol.
  153, Springer, Cham, 1969 (English).

\bibitem{GallowayMendes2025}
Gregory~J. Galloway and Abra{\~a}o Mendes, \emph{Some rigidity results for
  charged initial data sets}, Nonlinear Anal., Theory Methods Appl., Ser. A,
  Theory Methods \textbf{256} (2025), 9 (English), Id/No 113780.

\bibitem{Gibbons}
Gary~W. Gibbons, \emph{Some comments on gravitational entropy and the inverse
  mean curvature flow}, Classical Quantum Gravity \textbf{16} (1999), no.~6,
  1677--1687 (English).

\bibitem{Giusti}
Enrico Giusti, \emph{Minimal surfaces and functions of bounded variation},
  Monogr. Math., Basel, vol.~80, Birkh{\"a}user, Cham, 1984 (English).

\bibitem{GromovLawson}
Mikhael Gromov and H.~Blaine {Lawson, Jr.}, \emph{Positive scalar curvature and
  the {Dirac} operator on complete {Riemannian} manifolds}, Publ. Math., Inst.
  Hautes {\'E}tud. Sci. \textbf{58} (1983), 83--196 (English).

\bibitem{Gromov1996}
Misha Gromov, \emph{Positive curvature, macroscopic dimension, spectral gaps
  and higher signatures}, Functional analysis on the eve of the 21st century.
  Volume II. In honor of the eightieth birthday of I. M. Gelfand. Proceedings
  of a conference, held at Rutgers University, New Brunswick, NJ, USA, October
  24-27, 1993, Boston, MA: Birkh{\"a}user, 1996, pp.~1--213 (English).

\bibitem{Gromov2020}
\bysame, \emph{No metrics with {Positive} {Scalar} {Curvatures} on {Aspherical}
  5-{Manifolds}}, Preprint, {arXiv}:2009.05332 [math.{DG}], 2020.

\bibitem{Gromov2023}
\bysame, \emph{Four lectures on scalar curvature}, Perspectives in scalar
  curvature. In 2 volumes, Singapore: World Scientific, 2023, pp.~1--514
  (English).

\bibitem{HassScott}
Joel Hass and Peter Scott, \emph{The existence of least area surfaces in
  3-manifolds}, Trans. Am. Math. Soc. \textbf{310} (1988), no.~1, 87--114
  (English).

\bibitem{HuiskenPolden}
Gerhard Huisken and Alexander Polden, \emph{Geometric evolution equations for
  hypersurfaces}, Calculus of variations and geometric evolution problems.
  Lectures given at the 2nd session of the Centro Internazionale Matematico
  Estivo (CIME), Cetraro, Italy, June 15--22, 1996, Berlin: Springer, 1999,
  pp.~45--84 (English).

\bibitem{LesourdUngerYau}
Martin Lesourd, Ryan Unger, and Shing-Tung Yau, \emph{The positive mass theorem
  with arbitrary ends}, J. Differ. Geom. \textbf{128} (2024), no.~1, 257--293
  (English).

\bibitem{LimaSousaBatista}
Alexandre~B. Lima, Paulo~A. Sousa, and Rondinelle~M. Batista, \emph{Rigidity of
  marginally outer trapped surfaces in charged initial data sets}, Lett. Math.
  Phys. \textbf{115} (2025), no.~2, 15 (English), Id/No 41.

\bibitem{MarquesNeves}
Fernando~C. Marques and Andr{\'e} Neves, \emph{Rigidity of min-max minimal
  spheres in three-manifolds}, Duke Math. J. \textbf{161} (2012), no.~14,
  2725--2752 (English).

\bibitem{Mazet}
Laurent Mazet, \emph{Stable minimal hypersurfaces in $\mathbb{R}^6$}, Preprint,
  {arXiv}:2405.14676 [math.{DG}], 2024.

\bibitem{MeeksSimonYau}
William~H. Meeks~III, Leon Simon, and Shing-Tung Yau, \emph{Embedded minimal
  surfaces, exotic spheres, and manifolds with positive {Ricci} curvature},
  Ann. Math. (2) \textbf{116} (1982), 621--659 (English).

\bibitem{Mendes2019}
Abra{\~a}o Mendes, \emph{Rigidity of marginally outer trapped (hyper)surfaces
  with negative {{\(\sigma \)}}-constant}, Trans. Am. Math. Soc. \textbf{372}
  (2019), no.~8, 5851--5868 (English).

\bibitem{Mendes2025}
\bysame, \emph{Area-charge inequality and local rigidity in charged initial
  data sets}, Preprint, {arXiv}:2505.20060 [math.{DG}], 2025.

\bibitem{MicallefMoraru}
Mario Micallef and Vlad Moraru, \emph{Splitting of 3-manifolds and rigidity of
  area-minimising surfaces}, Proc. Am. Math. Soc. \textbf{143} (2015), no.~7,
  2865--2872 (English).

\bibitem{Nunes}
Ivaldo Nunes, \emph{Rigidity of area-minimizing hyperbolic surfaces in
  three-manifolds}, J. Geom. Anal. \textbf{23} (2013), no.~3, 1290--1302
  (English).

\bibitem{Simon}
Walter Simon, \emph{Bounds on area and charge for marginally trapped surfaces
  with a cosmological constant}, Classical Quantum Gravity \textbf{29} (2012),
  no.~6, 5 (English), Id/No 062001.

\bibitem{Wang}
Jian Wang, \emph{Topology of $3$-manifolds with uniformly positive scalar
  curvature}, Preprint, {arXiv}:2212.14383 [math.{DG}], 2022.

\bibitem{XinZhu}
Xin Zhou and Jonathan Zhu, \emph{Existence of hypersurfaces with prescribed
  mean curvature {I} -- generic min-max}, Camb. J. Math. \textbf{8} (2020),
  no.~2, 311--362 (English).

\bibitem{Zhu2021}
Jintian Zhu, \emph{Width estimate and doubly warped product}, Trans. Am. Math.
  Soc. \textbf{374} (2021), no.~2, 1497--1511 (English).

\bibitem{Zhu2023}
\bysame, \emph{Rigidity results for complete manifolds with nonnegative scalar
  curvature}, J. Differ. Geom. \textbf{125} (2023), no.~3, 623--644 (English).

\bibitem{Zhu2024ArXiv}
\bysame, \emph{Calabi-{Y}au type theorem for complete manifolds with
  nonnegative scalar curvature}, Preprint, {arXiv}:2402.15118 [math.{DG}],
  2024, pp.~1--14.

\bibitem{Zhu2024TAMS}
\bysame, \emph{Riemannian-{P}enrose inequality without horizon in dimension
  three}, Trans. Am. Math. Soc. \textbf{377} (2024), no.~6, 4101--4116
  (English).

\end{thebibliography}

\end{document}